\documentclass[preprint,12pt]{elsarticle}
\usepackage{url}

\usepackage{amssymb}
\usepackage{amsmath}
\usepackage{amsthm}
\usepackage{setspace}
\usepackage{mathrsfs}
\usepackage{booktabs}   
\usepackage{tabularx}
\usepackage{siunitx}    
\usepackage{geometry}   
\usepackage{caption}    
\usepackage{algorithm}
\usepackage{algorithmic} 
\usepackage{natbib}
\newtheorem{lemma}{Lemma}

\journal{Mathematics and Computers in Simulation}

\begin{document}

\begin{frontmatter}



\title{Closed-form solutions for VIX derivatives in a Legendre empirical model}


\author[math]{Yingli Wang\corref{cor1}}
\ead{2022310119@163.sufe.edu.cn}

\author[math]{Chenglong Xu\corref{cor2}}
\ead{xu.chenglong@shufe.edu.cn}

\author[math]{Ping He\corref{cor3}}
\ead{pinghe@mail.shufe.edu.cn}

\cortext[cor2]{Corresponding author}

\affiliation[math]{
  organization={School of Mathematics, Shanghai University of Finance and Economics},
  city={Shanghai},
  postcode={200433},
  country={China}
}
\begin{abstract}
  In this paper, we introduce a data-driven, single-parameter Markov diffusion model for the VIX. The volatility factor evolves in $(-1,1)$ with a uniform invariant distribution ensured by Legendre polynomials, mapped to the empirical distribution. We derive analytical series solutions for VIX futures and options using separation of variables to solve the Feynman-Kac PDE. Compared to the 3/2 model, our approach offers equal or superior accuracy and flexibility, providing an efficient, robust alternative for VIX pricing and risk management. Code and data are available at \url{github.com/gagawjbytw/empirical-VIX}.
\end{abstract}

\begin{keyword}
  the VIX index \sep closed-form solutions \sep numerical simulation \sep computational finance \sep mathematical modeling
  
  \MSC[2020] 91G20 \sep 60J25 \sep 65C30
  
  \end{keyword}

\end{frontmatter}



\section{Introduction}
Over the past thirty years, variance products and volatility trading have gained popularity with the introduction of the VIX index and its derivatives. The VIX is a volatility index calculated by the Chicago Board Options Exchange (CBOE), which informs investors about the expected market volatility of the S\&P 500 index in the next 30 calendar days. \citet{whaley2009understanding} provides a comprehensive explanation of the VIX index and its role in financial markets as an important barometer of investor sentiment. \citet{wang2019vix} shows that VIX has significant predictive power for international stock market volatility, with large VIX movements having particularly strong explanatory ability for market volatility forecasting. Furthermore, \citet{bekaert2014vix} decompose the squared VIX index into the conditional variance of stock returns and the equity variance premium, demonstrating that while the variance premium predicts stock returns, the conditional stock market variance is more effective at predicting economic activity and financial instability.

The modeling and simulation of volatility dynamics represent a challenging mathematical problem that requires sophisticated computational techniques. Our research aligns with the scope of Mathematics and Computers in Simulation by presenting a novel approach to simulating market dynamics through innovative mathematical modeling. The computational aspects of our work, particularly the development of efficient algorithms for pricing derivatives based on spectral methods, offer valuable insights into scientific computation in financial mathematics.

The VIX expresses volatility in percentage points and is derived under the assumption of zero jump on the process of the forward price ${F_{t,T},0\le t\le T}$ to the S\&P 500 index. Under a risk-neutral measure $\mathbb{Q}$, the dynamics of $F_{t,T}$ are governed by the equation 
\[
  \frac{dF_{t,T}}{F_{t,T}}=\sigma_{t,T}dW_t,
\]
where $\sigma_{t,T}$ is the stochastic volatility and ${W_t}$ is the standard $\mathbb{Q}$-Brownian motion. The square of the VIX index can be defined as 
\[
  {\rm VIX}_{t}^2:=\mathbb{E}^\mathbb{Q}\left[\frac1{T-t}\int_t^{T}\sigma_{s,T}^2ds\right].
\]
For further information on the VIX, we refer the reader to \citet{guoyuquan}.

With the growing importance of volatility trading, VIX derivatives have become increasingly important instruments for hedging and speculation. \citet{zhang2006vix} examine VIX futures contracts, proposing a stochastic variance model for VIX evolution and developing pricing expressions. Their empirical findings suggest that model calibration using recent data significantly improves pricing accuracy. Additionally, \citet{cheng2019vix} investigates the volatility premium embedded in VIX futures (the VIX premium), revealing a puzzling pattern where ex ante premiums fall or remain flat when risk measures increase, despite reliably predicting ex post returns to VIX futures.

Many scholars have focused on research related to the stochastic volatility models (SVM). The Heston model introduced in \citet{Heston} is a well-known model that assumes volatility dynamics to be a Cox-Ingersoll-Ross (CIR) process, which guarantees positivity and the mean-reverting property. The Bergomi model for future variance is introduced in \citet{Bergomi}. Additionally, there has been significant research on joint models for SPX and VIX options. The 3/2 model is a new model designed to fit both markets, which was first studied in \citet{Heston3to2} and has been analyzed by \citet{BB}, \citet{DF}. This model is a popular choice because it can reproduce the increasing right-hand implied-volatility skew in VIX options.

Our approach draws inspiration from mathematical modeling of dynamical systems, particularly the concept of equilibrium distributions and ergodicity. Just as dynamical systems tend to evolve towards their equilibrium states characterized by specific probability distributions, we propose that financial markets, specifically the VIX index, exhibit similar behavior. The bridge between mathematical modeling and computational simulation has been increasingly recognized in the literature. \citet{poitras2015classical} review the etymology and history of the classical ergodicity hypothesis and establish its connection to the fundamental empirical problem of using nonexperimental data to verify theoretical propositions in financial mathematics. In our model, we treat the VIX as a complex system whose dynamics can be described by methods from stochastic differential equations and numerical analysis, where the empirical distribution serves as the equilibrium state. This perspective allows us to leverage powerful computational tools, such as spectral methods based on orthogonal polynomials, to develop a more robust and computationally efficient approach to financial modeling.

The fast mean-reverting characteristics of volatility have been extensively studied in the literature. \citet{fouque1999financial} exploit the observed 'bursty' or persistent nature of stock price volatility, noting that volatility reverts slowly to its mean at high frequency but exhibits fast mean-reversion over the time scale of derivative contracts. Their asymptotic analysis yields simplified pricing and implied volatility formulas that effectively 'fit the skew' from European index options. Further extending this approach, \citet{cozma2020simulation} study systems of stochastic processes with fast mean-reverting volatilities, where coefficients containing fast volatility are replaced by ergodic averages—essentially applying a law of large numbers principle. \citet{hambly2020fast} develop this concept for large portfolio models, showing that in a fast mean-reverting volatility environment, an approximate constant volatility model can accurately estimate the distribution of losses. Additionally, \citet{dragulescu2002probability} provide analytical solutions for the time-dependent probability distribution of stock price changes under the Heston stochastic volatility model, demonstrating excellent agreement with empirical data across multiple time scales.

Regarding the valuation of VIX derivatives, \citet{mencia2013valuation} conduct an extensive empirical analysis of various pricing models before, during, and after the 2008-2009 financial crisis. Their findings suggest that a process for the log of the observed VIX combining central tendency and stochastic volatility reliably prices VIX derivatives, while also revealing a significant risk premium that affects the long-run volatility level. In a complementary study, \citet{zhu2012analytical} present a closed-form, exact solution for pricing VIX futures in a stochastic volatility model with simultaneous jumps in asset price and volatility processes, demonstrating that the Heston stochastic volatility model performs well for VIX futures pricing.

In actual operation, multi-parameter calibration is sometimes a difficult thing. When optimizing, there are often no solutions or unstable situations. For example, the well-known 3/2 model, the calibration usually involves three parameters and it is usually difficult to calibrate, for example, \citet{gudmundsson2019calibration}.

In this paper, we propose a new model that relaxes the constant reverting assumption and instead assumes that the volatility process converges to its invariant distribution. In the theory of Markov processes, a Markov process that has an invariant distribution will converge to that distribution under certain conditions. By assuming that the VIX process is a Markov process and that its empirical distribution is close to its invariant distribution, our model can capture the dynamics of the VIX index more accurately than traditional models. This idea is inspired by the Metropolis-Hastings algorithm, which generates a sequence of samples from a target distribution by constructing a Markov chain whose stationary distribution is the target distribution. If the proposal distribution is designed appropriately, the Metropolis-Hastings algorithm can converge to the desired invariant distribution. Another famous statistical method that uses this idea is the Gibbs sampler, which is a special case of the Metropolis-Hastings algorithm where the proposal distribution is replaced with the conditional distributions of each variable given the other variables. For more information on the Metropolis-Hastings algorithm and the Gibbs sampler, we refer the reader to \citet{SAGX}. 

The question we face now is how to construct a continuous Markov process that closely matches the empirical distribution of the VIX index. In this paper, we propose using Legendre polynomials as eigenfunctions of the generator for a Markovian diffusion process ${X_t}$. We have identified a model whose invariant distribution closely approximates the empirical distribution of the VIX. Thanks to a spectral expansion of the pricing formula based on Legendre polynomials, calculations are both accurate and efficient. Moreover, this model requires the calibration of only a single parameter, $\kappa$. We show that if the invariant distribution of the process is a uniform distribution on $[-1,1]$, which can be achieved with Legendre polynomials, then we can accurately model the VIX index based on its historical empirical distribution.

\section{Motivation and modeling methods}
 We analyze daily data from 1990/01/02 to 2024/12/30 and construct a continuous-time Markov process $h(X_t)$ to model the VIX index. Here, $(X_t)_{t\ge0}$ is a diffusion process with an invariant distribution of $U[-1,1]$, and $h$ is a function that ensures the distribution of $h(2Y-1)$ matches the empirical distribution of the data, where $Y$ has a distribution of $U[0,1]$. We observe that if $h$ is a monotone increasing function, it is equivalent to the inverse function of the empirical CDF. Finally, we apply dynamic analyses to investigate the pricing of the corresponding derivatives and present numerical solutions of the pricing problem using the method of separation of variables.

In this paper, we assume that the volatility is determined by the factor process $(X_t)_{t\ge0}$, which is a diffusion with double-side barriers given by
\begin{equation}\label{doublebarrier}
dX_t=-\kappa X_tdt+\sqrt{\kappa(1-X_t^2)}dW_t,
\end{equation}
where $\kappa>0$ (can be calibrated from the VIX option data) and $X_0\in(-1,1)$, $(W_t)_{t\ge0}$ is a Brownian motion under risk-neutral probability measure. The process $(X_t)_{t\ge0}$ can only take values in the interval $(-1,1)$. We can write the generator of $(X_t)_{t\ge0}$ as
\[ 
  \mathcal{L}=-\kappa x\frac{\partial}{\partial x}+\frac{\kappa}{2}(1-x^2)\frac{\partial^2}{\partial x^2}.
\]
\begin{lemma}\label{lemma1}
The process $(X_t)_{t\ge0}$ determined by \eqref{doublebarrier} has a unique invariant distribution $U[-1,1]$.\\

\end{lemma}
\begin{proof}
It is easy to see that $\mathcal{L}$ is a self-adjoint operator. \citet{Hairer} give the criterion for the existence of the invariant probability measure, whether there is a $u\ge0$ such that $\mathcal{L}u\le C_1-C_2u$ for some constants $C_1,C_2>0$. Let us try $u=1$, and $C_2=\frac12C_1>0$, then $\mathcal{L}u+C_2u=C_2\le C_1$.

Now we can give the invariant distribution as follows. The forward equation of the process is 
\[
 \frac{\partial f}{\partial t}=-\kappa x\frac{\partial f}{\partial x}+\frac{\kappa}{2}(1-x^2)\frac{\partial^2f}{\partial x^2}.
\]
Set $g(x)$ the probability density function of the invariant distribution, we obtain 
\[
  -\kappa xg'+\frac{\kappa}{2}(1-x^2)g''=0.
\]
Obviously, any constant function is a solution to the equation. Since $g(x)$ is a probability density function on $[-1,1]$, so $g(x)=1/2$. We claim that the non-constant solution can't be the pdf of the invariant distribution. In fact, the non-constant solution can be given by $g(x)=c\log\frac{1+x}{1-x}$, where $c$ is a constant. However, it may takes negative values, thus it can't be a probability density function.
\end{proof}
Assume that $h$ is a function lies in $[0,1]$, we can establish a model from the VIX factor $X=(X_t)_{t\ge0}$, i.e., there is a function $h$ such that 
\[
  {\rm VIX}_t:=h\left(\frac12(X_t+1)\right).
\]
We try to find the appropriate function $h$ to describe the relationship between $(X_t)_{t\ge0}$ and $({\rm VIX}_t)_{t\ge0}$. Given $Y\sim U[0,1]$, the problem comes down to find the function $h$ such that the distribution of $h(Y)$ is close to the empirical distribution based on the historical data of VIX. Let $F$ be the cumulative distribution function of the empirical distribution based on the historical data of VIX. Suppose $h$ is monotone increasing and $Z\sim U[-1,1]$, then for a given real number $c$, 
\[ 
  F(c)=\mathbb{P}\left(h\left(\frac12(Z+1)\right)\le c\right)=\mathbb{P}\left(\frac12(Z+1)\le h^{-1}(c)\right)=h^{-1}(c).
\]
So $h=F^{-1}$ is exactly the function we need. Intuitively, the parameter $\kappa$ determines the speed at which the VIX reverts to its stationary distribution.

\begin{figure}[H]
  \centering
  \begin{minipage}{0.32\textwidth}
    \centering
    \includegraphics[width=\textwidth]{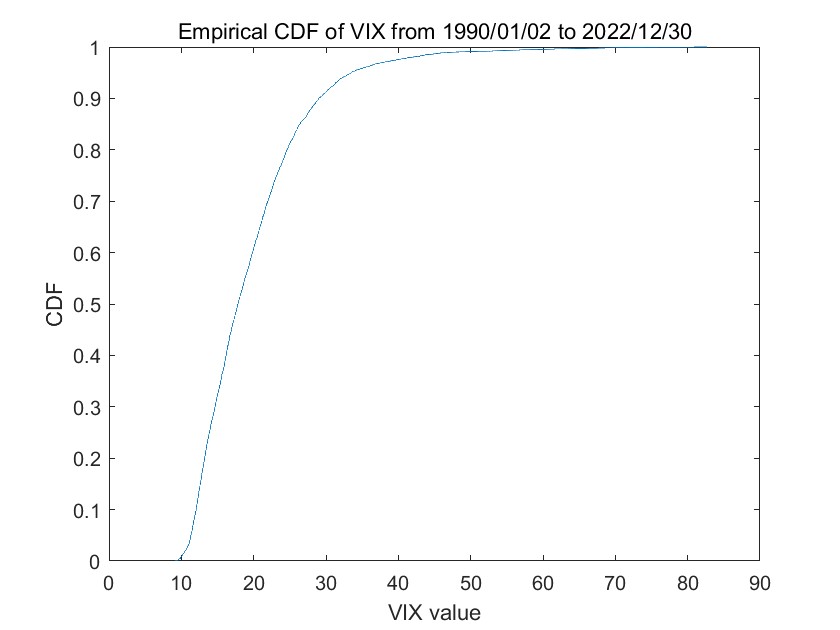}
    \caption{Empirical CDF of VIX index.}
    \label{fig:empiricalCDF}
  \end{minipage}
  \hfill
  \begin{minipage}{0.32\textwidth}
    \centering
    \includegraphics[width=\textwidth]{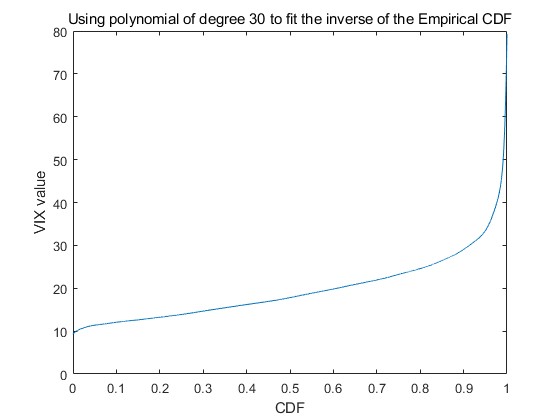}
    \caption{The function $h$ fit by 30 degree polynomial.}
    \label{fig:degree30}
  \end{minipage}
  \hfill
  \begin{minipage}{0.32\textwidth}
    \centering
    \includegraphics[width=\textwidth]{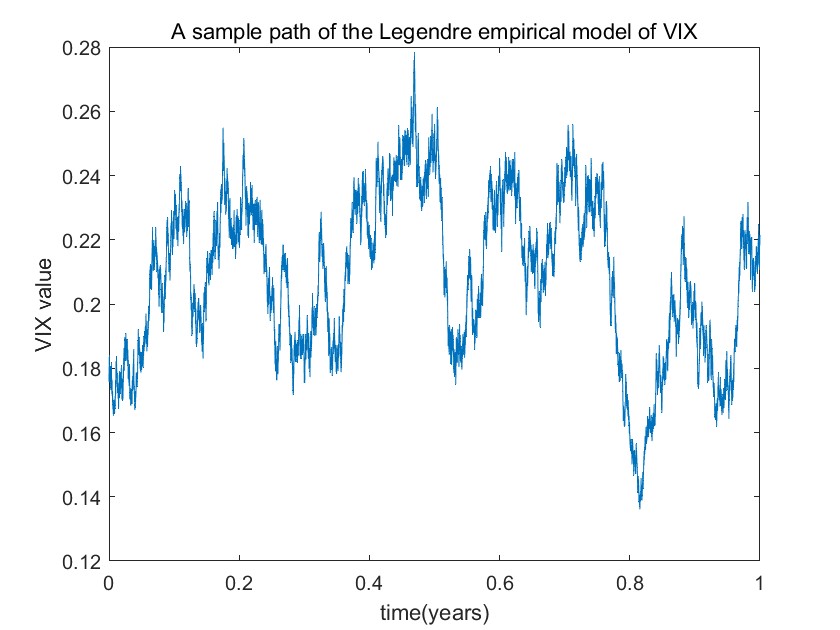}
    \caption{A sample path of the empirical model of VIX, where $\kappa=2.362$.}
    \label{fig:samplepath}
  \end{minipage}
\end{figure}

\section{Pricing of the corresponding VIX derivatives}
Now, let's shift our attention to the pricing of VIX derivatives. Specifically, we need to determine a reasonable price for VIX futures and VIX call/put options given the current value of the VIX index. To achieve this, we need to obtain the value of $X_t$ from the value of ${\rm VIX}_t$. To do so, we can use polynomial fitting and solve the resulting polynomial equations. 

In this section, we consider the pricing of the corresponding VIX derivatives. Let $p(t,x,y)$ denote the transition probability density of $(X_t)_{t\ge0}$, which satisfies
\begin{equation}\label{forwardequation}
  \begin{cases}
    \frac{\partial p}{\partial t}=\frac12\kappa(1-y^2)\frac{\partial^2p}{\partial y^2}-\kappa y\frac{\partial p}{\partial y},\ t\in[0,T),y\in(-1,1)\\
    p(0,x,y)=\delta(y-x).
  \end{cases}
\end{equation}
Then we can easily give the pricing formulas.
\subsection{Pricing VIX futures and call/put options}
The price of VIX futures is given by 
\begin{equation}\label{VIXfutures}
  F(t,x)
  :=\mathbb{E}[{\rm VIX}_T|X_t=x]
  =\int_{-1}^1h(\frac12(y+1))p(T-t,x,y)dy
  =\mathbb{E}^x\left[h(\frac12(X_{T-t}+1))\right]\ ,t\in[0,T],
\end{equation}
where $\mathbb{E}^x$ represents the expectation under the probability measure $\mathbb{P}(\cdot|X_0=x)$. VIX options settle by cash and trade in European style. Due to the Markov property of our model, European call options maturing at $T$ with strike price $K$ are given by
\begin{align*}
c(t,x) 
= e^{-r(T-t)}\mathbb{E}[({\rm VIX}_T-K)^+|X_t=x]
= e^{-r(T-t)}\mathbb{E}^x[(h(\tfrac12(1+X_{T-t}))-K)^+].
\end{align*}
We aim to use the method of separation of variables to obtain series solutions for pricing VIX futures and call options. First, we consider VIX futures. Using the Feynman-Kac formula, price of VIX future $F(t,x)$ in \eqref{VIXfutures} satisfies the Cauchy problem
$$
\begin{cases}
\frac{\partial F}{\partial t}=\kappa x\frac{\partial F}{\partial x}-\frac{\kappa}{2}(1-x^2)\frac{\partial^2F}{\partial x^2}, & t\in[0,T],\\
F(T,x)=h(\tfrac12(x+1)).
\end{cases}
$$
By the method of separation of variables, we set $F(t,x)=X(x)T(t)$, which gives
$$
X(x)T'(t)=\kappa xX'(x)T(t)-\frac{\kappa}{2}(1-x^2)X''(x)T(t).
$$
Utilizing the method of separation of variables by setting $F (t, x)=X(x)T(t)$, we have:

$$\frac{T'(t)}{\kappa T(t)}=\frac{xX'(x)-\frac12(1-x^2)X''(x)}{X(x)}=\lambda,$$
where $\lambda = \frac 12 n(n+1)$. The equation about $X(x)$ is the Legendre equation, which has solutions in the form of Legendre polynomials of order $n$, i.e., $X(x)=P_n(x)$. The equation about $T(t)$ is a simple ODE with solution:
\[
  T(t)=c_n\exp\left(\frac{\kappa n(n+1)t}{2}\right).
\]
Thus, our solution for $F(t,x)$ becomes:

$$F(t,x)=\sum_{n=0}^\infty c_n \exp\left(\frac{\kappa n(n+1)}{2}t\right) P_n(x).$$

Using the Legendre polynomial orthogonality and substituting in the terminal conditions, we derive the coefficient expression:

$$c_n\exp\left(\frac{\kappa n(n+1)T}{2}\right)=\frac{2n+1}{2}\int_{-1}^1h\left(\frac12(x+1)\right)P_n(x)dx.$$

Assuming $h$ is a polynomial of degree 30, we can express it under the Legendre basis yielding an easy calculation for the inner product instead of using numerical integration. From the result of fitting, $h$ is indeed a monotone increasing function, see Figure \ref{fig:degree30}. This gives us the final expression of $F(t,x)$ as:

$$F(t,x)=\sum_{n=0}^{30} \frac{2n+1}{2}\langle \tilde{h},P_n \rangle \exp\left(-\frac12\kappa n(n+1)(T-t)\right)P_n(x),$$

where $\tilde{h}(x)= h(\frac{1}{2}(x+1))$ and $\langle f,g\rangle=\int_{-1}^1f(x)g(x)dx.$

Similarly, considering a European call option on VIX with strike price $K$, we set $\tilde{h}_1=(\tilde{h}-K)^+$, and obtain the expression for the option pricing as:
\begin{equation}
c(t,x)=e^{-r(T-t)}\sum_{n=0}^\infty \frac{2n+1}{2}\langle\tilde{h}_1,P_n\rangle
\exp\left(-\frac12\kappa n(n+1)(T-t)\right)P_n(x).
\end{equation}

\subsection{The estimation of calculation error}
The table below displays the prices of VIX call options based on the empirical model, given various initial values of the VIX index. We employ a fixed polynomial function $h$ of degree 30 to calculate the corresponding value of $X_t$ from the observed value of ${\rm VIX}_t$. To determine the appropriate number of terms for the polynomial, we compare the results obtained by summing 6, 11, 21, and 31 terms, and find that using 31 terms yields accurate prices.

\begin{table}[h]
  \centering
  \caption{The VIX call option price numerical error estimation, $K=0.2$, $r=0.05$, $T=1/6$, $t=1/12$. The parameter $\kappa=2.362$.}
  \label{tab:vix-call-errors}
  \begin{tabular}{lcccc}
  \toprule
  \textbf{${\rm VIX}_t$} & \textbf{6 terms} & \textbf{11 terms} & \textbf{21 terms} & \textbf{31 terms} \\
  \midrule
  0.1 & 0.0010 & 0.0024 & 0.0024 & 0.0024 \\
  0.3 & 0.1004 & 0.0991 & 0.0991 & 0.0991 \\
  0.5 & 0.1849 & 0.1871 & 0.1871 & 0.1870 \\
  0.7 & 0.1952 & 0.1984 & 0.1984 & 0.1984 \\
  \bottomrule
  \end{tabular}
\end{table}

\section{The inverse problem for $\kappa$}
To make the model adaptable to different time horizons and fit the market option data, we need to solve an inverse problem for $\kappa$, which involves finding a positive constant $\kappa$ that minimizes the distance between the theoretical values and the practical market values. Mathematically, we can express this as 
\[
  \hat{\kappa}=\arg\min_{\kappa\in[a,b]}\sum_{i=1}^N(c(t_i,x_i;\kappa,T)-\tilde{c}_i)^2,
\]
where $[a,b]$ is a guessing interval which can be determined empirically (Based on our experience, it is more appropriate to set $a$ to 3.5 and $b$ to 4.5). Here, $N$ represents the number of continuous observations, $x_i$ is the factor process value at time $t_i$, which needs to be solved by $h(\frac12(x_i+1))={\rm VIX}_{t_i}$, and $\tilde{c}_i$ is the observed call option value at time $t_i$.

Based on the previous statement, we can write the above equation as:
\begin{equation}
  \resizebox{\columnwidth}{!}{$
    \begin{aligned}
      \hat{\kappa}
      =&\arg\min_{\kappa\in[a,b]}\sum_{i=1}^N\bigg(e^{-r(T-t_i)}\sum_{n=0}^\infty \frac{2n+1}{2}\langle\tilde{h}_1,P_n\rangle
      \exp\left(-\frac{1}{2}\kappa n(n+1)(T-t_i)\right)\times P_n(x_i)-\tilde{c}_i\bigg)^2\\
      \approx&\arg\min_{\kappa\in[a,b]}\sum_{i=1}^N\left(e^{-r(T-t_i)}\sum_{n=0}^{30}\tilde{\nu}_{i,n}\exp\left(-\frac{1}{2}n(n+1)(T-t_i)\kappa\right)-\tilde{c}_i\right)^2,
    \end{aligned}
  $}
\end{equation}
where $\tilde{\nu}_{i,n}=\frac{2n+1}{2}\langle \tilde{h}_1,P_n \rangle P_n(x_i)$, $i=1,2,\cdots,N$.

Unlike other models, our approach requires calibrating only a single parameter, $\hat{\kappa}$. This equation allows us to find the optimal value of $\kappa$ that minimizes the difference between the theoretical and practical market call option prices. It is a one-dimensional optimization problem that can be easily solved numerically. To simplify the calculation, we can use a finite sum with 30 terms instead of an infinite sum. Then, we can apply standard optimization techniques, such as gradient-based methods or derivative-free methods, to solve for the optimal value of $\kappa$. The solution for $\hat{\kappa}$ will give us the value of $\kappa$ that minimizes the distance between the theoretical value and the practical value, and hence provides the best fit for the option data.

\section{Empirical study results and comparison with 3/2 model}
In this section, we focus on comparing our Legendre model with the 3/2 model, while excluding other models such as the Bergomi model from the comparison. This choice is motivated by the fact that the 3/2 model and the Bergomi model are currently the two mainstream models for VIX options. However, the Bergomi model's calibration relies on different data inputs compared to the 3/2 model, primarily requiring the S\&P 500 implied volatility surface to construct forward variance curves. Since our Legendre model shares the same data input requirements as the 3/2 model, it is more appropriate and meaningful to conduct a direct comparison between these two models.

Let us first introduce the well-known 3/2 model. If we write the dynamics of 3/2 model as
\begin{equation}\label{3to2dynamics}
  dV_t=(\alpha V_t+\beta V_t^2)dt+k V_t^{3/2}d\tilde{W}_t,
\end{equation}
where $(\tilde{W}_t)_{t\ge0}$ is the Brownian motion under risk-neutral measure. \citet{3to2option} gave the analytic solutions for call option prices on the VIX under the 3/2-model as follows.
\begin{lemma}[Goard \& Mazur(2013)]
  The value of a call option on the {\rm VIX} with strike $K$ and expiry $T$, when the ${\rm VIX}$ prices $V$, follow the risk-neutral process (\ref{3to2dynamics}) with $\beta<0$, is given by
  \[
  \begin{aligned}
  C(V,t)
  =&\frac{2\alpha e^{-r(T-t)}}{k^2p}\exp\left(\frac{-2\alpha e^{-\alpha(T-t)}}{k^2Vp}\right)V^{-\frac{\beta}{k^2}+\frac12}\times\exp\left(\alpha(T-t)(-\frac{\beta}{k^2}+\frac12)\right)\\
  &\times\int_0^{1/K}u^{\frac12-\frac{\beta}{k^2}}\left(\frac1u-K\right)e^{-\frac{2\alpha u}{k^2p}}I_\nu\left(\frac{4\alpha \sqrt{u}e^{-\frac{\alpha(T-t)}{2}}}{k^2\sqrt{V}p}\right)du,
  \end{aligned}
\]
where $\nu=1-\frac{2\beta}{k^2},p=1-\exp(-\alpha(T-t))$ and $I_\nu(\cdot)$ is the modified Bessel function of order $\nu$.
\end{lemma}

We analyze four recent options with a strike price of 20, expiring on 2024/10/30, 2024/11/20, 2024/12/18, and 2024/12/24. After calibration, we compute the average relative errors for both the Legendre model and the 3/2 model.

\begin{figure}[H]
  \centering
  
  \begin{minipage}[t]{0.48\columnwidth}
    \centering
    \includegraphics[width=\linewidth]{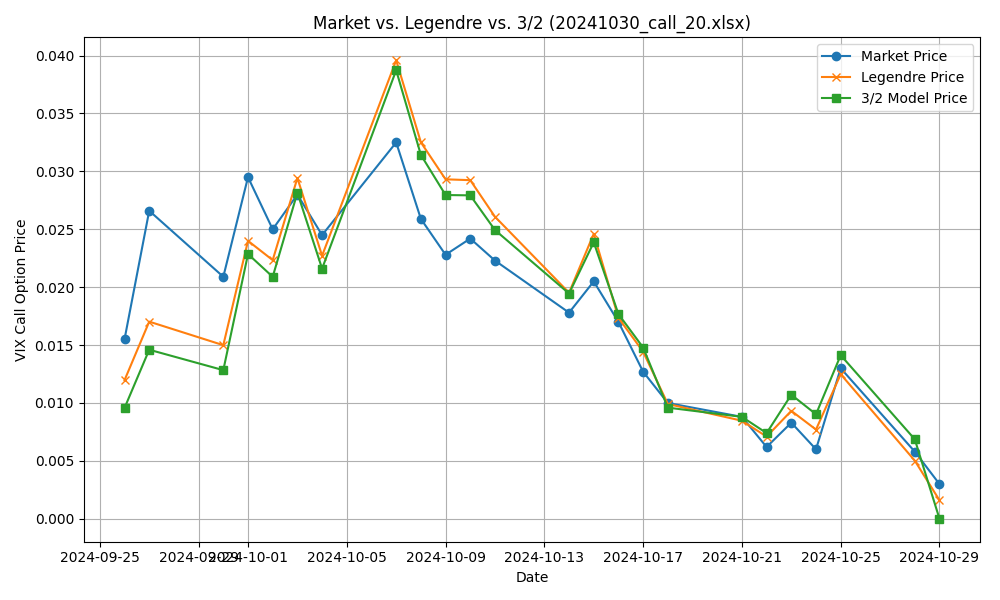}
    \caption{Market vs. Legendre vs. 3/2 (expired on 2024/10/30)}
    \label{fig:20241030expired}
  \end{minipage}
  \hfill
  \begin{minipage}[t]{0.48\columnwidth}
    \centering
    \includegraphics[width=\linewidth]{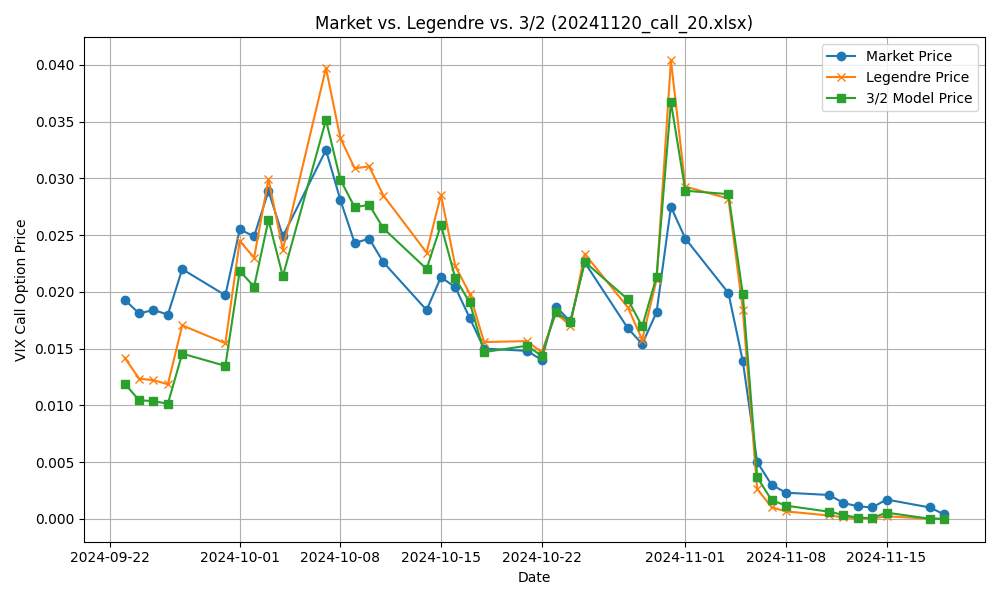}
    \caption{Market vs. Legendre vs. 3/2 (expired on 2024/11/20)}
    \label{fig:20241120expired}
  \end{minipage}
  
  \vspace{0.5cm} 
  
  \begin{minipage}[t]{0.48\columnwidth}
    \centering
    \includegraphics[width=\linewidth]{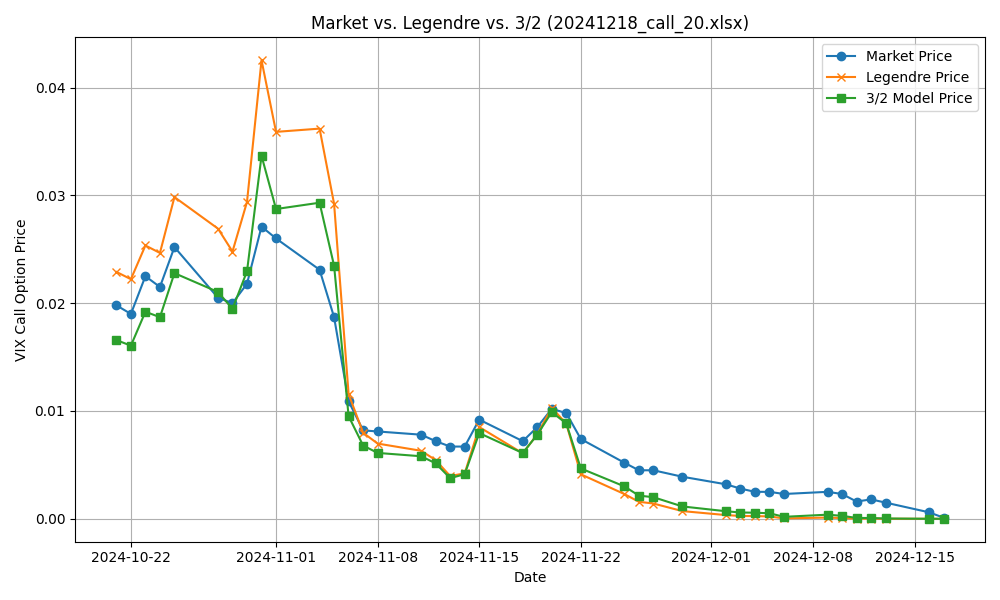}
    \caption{Market vs. Legendre vs. 3/2 (expired on 2024/12/18)}
    \label{fig:20241218expired}
  \end{minipage}
  \hfill
  \begin{minipage}[t]{0.48\columnwidth}
    \centering
    \includegraphics[width=\linewidth]{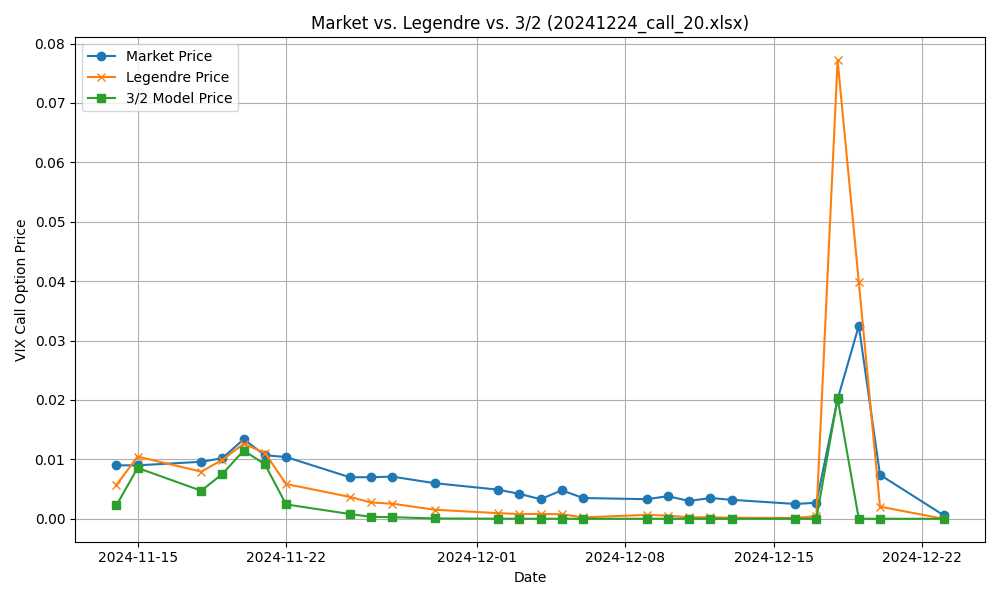}
    \caption{Market vs. Legendre vs. 3/2 (expired on 2024/12/24)}
    \label{fig:20241224expired}
  \end{minipage}
  
  \caption{Comparison of Average Relative Errors for Legendre and 3/2 Models}
  \label{fig:avg-errors}
\end{figure}

\begin{table}[h]
  \centering
  \caption{Comparison of Average Relative Errors for Legendre and 3/2 Models}
  \label{tab:avg-errors}
    \begin{tabular}{lcccc}
      \toprule
      \textbf{Model} & \textbf{2024/10/30} & \textbf{2024/11/20} & \textbf{2024/12/18} & \textbf{2024/12/24} \\
      \midrule
      Legendre                & 17.15\% & 34.02\% & 49.47\% & 69.97\% \\
      3/2 Model               & 22.49\% & 31.02\% & 42.25\% & 79.28\% \\
      \bottomrule
    \end{tabular}
\end{table}

\section{Conclusion}
In this paper, we have developed a novel single-parameter Markov diffusion model for the VIX index, leveraging Legendre polynomials to ensure a uniform invariant distribution mapped to the empirical distribution. By employing the method of separation of variables, we derived closed-form solutions for pricing VIX futures and options, which demonstrate high accuracy and computational efficiency. Our numerical approach offers a practical implementation of advanced mathematical techniques in financial modeling, combining spectral methods with stochastic differential equations to solve real-world pricing problems. The computational experiments, based on historical VIX data from 1990 to 2024, demonstrate that our proposed mathematical model achieves comparable or superior performance to the more complex 3/2 model, with the significant advantage of requiring only one parameter, $\kappa$, for calibration. This simplicity enhances the model's robustness and applicability in practical settings, exemplifying how mathematical modeling and scientific computation can be effectively applied to financial systems. The results, supported by numerical error estimations and market data comparisons, confirm the model's effectiveness for VIX derivatives pricing and risk management. Future work could explore extensions to incorporate multi-factor dynamics or apply this computational framework to other financial indices and derivatives markets.

\section{Acknowledgments}
The first author Yingli Wang is supported by the Fundamental Research Funds for the Central Universities in Shanghai University of Finance and Economics CXJJ-2023-397.

\section{Disclosure of Interest}
The authors declare that there are no conflicts of interest to disclose.

\section{Data Availability Statement}
The data and code supporting the findings of this study are publicly available at \url{github.com/gagawjbytw/empirical-VIX}.

\bibliographystyle{elsarticle-num-names} 
\bibliography{references}

\begin{thebibliography}{22}
\expandafter\ifx\csname natexlab\endcsname\relax\def\natexlab#1{#1}\fi
\providecommand{\url}[1]{\texttt{#1}}
\providecommand{\href}[2]{#2}
\providecommand{\path}[1]{#1}
\providecommand{\DOIprefix}{doi:}
\providecommand{\ArXivprefix}{arXiv:}
\providecommand{\URLprefix}{URL: }
\providecommand{\Pubmedprefix}{pmid:}
\providecommand{\doi}[1]{\href{http://dx.doi.org/#1}{\path{#1}}}
\providecommand{\Pubmed}[1]{\href{pmid:#1}{\path{#1}}}
\providecommand{\bibinfo}[2]{#2}
\ifx\xfnm\relax \def\xfnm[#1]{\unskip,\space#1}\fi
\bibitem[{Whaley(2009)}]{whaley2009understanding}
\bibinfo{author}{R.~E. Whaley},
\newblock \bibinfo{title}{Understanding the vix},
\newblock \bibinfo{journal}{Journal of Portfolio Management} \bibinfo{volume}{35} (\bibinfo{year}{2009}) \bibinfo{pages}{98--105}.
\bibitem[{Wang(2019)}]{wang2019vix}
\bibinfo{author}{H.~Wang},
\newblock \bibinfo{title}{{VIX and volatility forecasting: A new insight}},
\newblock \bibinfo{journal}{Physica A: Statistical Mechanics and its Applications} \bibinfo{volume}{533} (\bibinfo{year}{2019}) \bibinfo{pages}{121951}.
\bibitem[{Bekaert and Hoerova(2014)}]{bekaert2014vix}
\bibinfo{author}{G.~Bekaert}, \bibinfo{author}{M.~Hoerova},
\newblock \bibinfo{title}{The vix, the variance premium and stock market volatility},
\newblock \bibinfo{journal}{Journal of econometrics} \bibinfo{volume}{183} (\bibinfo{year}{2014}) \bibinfo{pages}{181--192}.
\bibitem[{Kwok and Zheng(2022)}]{guoyuquan}
\bibinfo{author}{Y.~Kwok}, \bibinfo{author}{W.~Zheng}, \bibinfo{title}{{Pricing Models of Volatility Products and Exotic Variance Derivatives}}, \bibinfo{edition}{first edition} ed., \bibinfo{publisher}{CRC Press}, \bibinfo{year}{2022}.
\bibitem[{Zhang and Zhu(2006)}]{zhang2006vix}
\bibinfo{author}{J.~E. Zhang}, \bibinfo{author}{Y.~Zhu},
\newblock \bibinfo{title}{Vix futures},
\newblock \bibinfo{journal}{Journal of Futures Markets: Futures, Options, and Other Derivative Products} \bibinfo{volume}{26} (\bibinfo{year}{2006}) \bibinfo{pages}{521--531}.
\bibitem[{Cheng(2019)}]{cheng2019vix}
\bibinfo{author}{I.-H. Cheng},
\newblock \bibinfo{title}{The vix premium},
\newblock \bibinfo{journal}{The Review of Financial Studies} \bibinfo{volume}{32} (\bibinfo{year}{2019}) \bibinfo{pages}{180--227}.
\bibitem[{Heston(1993)}]{Heston}
\bibinfo{author}{S.~Heston},
\newblock \bibinfo{title}{{A closed-form solution for options with stochastic volatility with applications to bond and currency options}},
\newblock \bibinfo{journal}{Review of Financial Studies} \bibinfo{volume}{6} (\bibinfo{year}{1993}) \bibinfo{pages}{327--343}.
\bibitem[{Bergomi(2015)}]{Bergomi}
\bibinfo{author}{L.~Bergomi}, \bibinfo{title}{Stochastic volatility modeling}, \bibinfo{publisher}{CRC Press}, \bibinfo{year}{2015}.
\bibitem[{Heston(1997)}]{Heston3to2}
\bibinfo{author}{S.~Heston}, \bibinfo{title}{{A Simple New Formula for Options with Stochastic Volatility}}, \bibinfo{year}{1997}. \bibinfo{note}{Manuscript, John M. Olin, School of Business, Washington University}.
\bibitem[{Baldeaux and Badran(2014)}]{BB}
\bibinfo{author}{J.~Baldeaux}, \bibinfo{author}{A.~Badran},
\newblock \bibinfo{title}{{Consistent modelling of VIX and equity derivatives using a 3/2 plus jumps model}},
\newblock \bibinfo{journal}{Applied Mathematical Finance} \bibinfo{volume}{21} (\bibinfo{year}{2014}) \bibinfo{pages}{299--312}.
\bibitem[{Drimus and Farkas(2013)}]{DF}
\bibinfo{author}{G.~Drimus}, \bibinfo{author}{W.~Farkas},
\newblock \bibinfo{title}{{Local volatility of volatility for the VIX market}},
\newblock \bibinfo{journal}{Review of Derivatives Research} \bibinfo{volume}{16} (\bibinfo{year}{2013}) \bibinfo{pages}{267--293}.
\bibitem[{Poitras and Heaney(2015)}]{poitras2015classical}
\bibinfo{author}{G.~Poitras}, \bibinfo{author}{J.~Heaney},
\newblock \bibinfo{title}{Classical ergodicity and modern portfolio theory},
\newblock \bibinfo{journal}{Chinese Journal of Mathematics} \bibinfo{volume}{2015} (\bibinfo{year}{2015}) \bibinfo{pages}{737905}.
\bibitem[{Fouque et~al.(1999)Fouque, Papanicolaou, and Sircar}]{fouque1999financial}
\bibinfo{author}{J.-P. Fouque}, \bibinfo{author}{G.~Papanicolaou}, \bibinfo{author}{K.~R. Sircar},
\newblock \bibinfo{title}{Financial modeling in a fast mean-reverting stochastic volatility environment},
\newblock \bibinfo{journal}{Asia-Pacific Financial Markets} \bibinfo{volume}{6} (\bibinfo{year}{1999}) \bibinfo{pages}{37--48}.
\bibitem[{Cozma and Reisinger(2020)}]{cozma2020simulation}
\bibinfo{author}{A.~S. Cozma}, \bibinfo{author}{C.~Reisinger},
\newblock \bibinfo{title}{Simulation of conditional expectations under fast mean-reverting stochastic volatility models},
\newblock in: \bibinfo{booktitle}{International Conference on Monte Carlo and Quasi-Monte Carlo Methods in Scientific Computing}, \bibinfo{organization}{Springer}, \bibinfo{year}{2020}, pp. \bibinfo{pages}{223--240}.
\bibitem[{Hambly and Kolliopoulos(2020)}]{hambly2020fast}
\bibinfo{author}{B.~Hambly}, \bibinfo{author}{N.~Kolliopoulos},
\newblock \bibinfo{title}{Fast mean-reversion asymptotics for large portfolios of stochastic volatility models},
\newblock \bibinfo{journal}{Finance and Stochastics} \bibinfo{volume}{24} (\bibinfo{year}{2020}) \bibinfo{pages}{757--794}.
\bibitem[{Dragulescu and Yakovenko(2002)}]{dragulescu2002probability}
\bibinfo{author}{A.~A. Dragulescu}, \bibinfo{author}{V.~M. Yakovenko},
\newblock \bibinfo{title}{Probability distribution of returnsin the heston model with stochastic volatility},
\newblock \bibinfo{journal}{Quantitative finance} \bibinfo{volume}{2} (\bibinfo{year}{2002}) \bibinfo{pages}{443}.
\bibitem[{Mencia and Sentana(2013)}]{mencia2013valuation}
\bibinfo{author}{J.~Mencia}, \bibinfo{author}{E.~Sentana},
\newblock \bibinfo{title}{Valuation of vix derivatives},
\newblock \bibinfo{journal}{Journal of Financial Economics} \bibinfo{volume}{108} (\bibinfo{year}{2013}) \bibinfo{pages}{367--391}.
\bibitem[{Zhu and Lian(2012)}]{zhu2012analytical}
\bibinfo{author}{S.-P. Zhu}, \bibinfo{author}{G.-H. Lian},
\newblock \bibinfo{title}{An analytical formula for vix futures and its applications},
\newblock \bibinfo{journal}{Journal of Futures Markets} \bibinfo{volume}{32} (\bibinfo{year}{2012}) \bibinfo{pages}{166--190}.
\bibitem[{Gudmundsson and Vyncke(2019)}]{gudmundsson2019calibration}
\bibinfo{author}{H.~Gudmundsson}, \bibinfo{author}{D.~Vyncke},
\newblock \bibinfo{title}{On the calibration of the 3/2 model},
\newblock \bibinfo{journal}{European Journal of Operational Research} \bibinfo{volume}{276} (\bibinfo{year}{2019}) \bibinfo{pages}{1178--1192}.
\bibitem[{Brooks et~al.(2011)Brooks, Gelman, Jones, and Meng}]{SAGX}
\bibinfo{author}{S.~Brooks}, \bibinfo{author}{A.~Gelman}, \bibinfo{author}{G.~Jones}, \bibinfo{author}{X.-L. Meng}, \bibinfo{title}{{Handbook of Markov Chain Monte Carlo}}, \bibinfo{publisher}{CRC Press}, \bibinfo{year}{2011}.
\bibitem[{Hairer(nd)}]{Hairer}
\bibinfo{author}{M.~Hairer}, \bibinfo{title}{{Convergence of Markov processes}}, \bibinfo{year}{n.d.} \bibinfo{note}{Lecture Notes, University of Warwick.}
\bibitem[{Goard and Mazur(2013)}]{3to2option}
\bibinfo{author}{J.~Goard}, \bibinfo{author}{M.~Mazur},
\newblock \bibinfo{title}{{Stochastic volatility models and pricing of VIX options}},
\newblock \bibinfo{journal}{Mathematical Finance} \bibinfo{volume}{23} (\bibinfo{year}{2013}) \bibinfo{pages}{439--458}.

\end{thebibliography}
\section{Appendix: Pseudocode for Calibration and Error Comparison}
\begin{algorithm}[H]
  \caption{Legendre + 3/2 Model Pricing Workflow (Example)}
  \label{alg:legendre_3_2}
  {\fontsize{8}{7}\selectfont
  \begin{algorithmic}[1]
    \STATE \textbf{Step 1: Read and fit historical VIX data (Legendre part)}
    \STATE \quad 1. Read Excel file and obtain \(\texttt{vix\_data}\).
    \STATE \quad 2. Sort data as \(\texttt{vix\_sorted}\) of length \(N_{\mathrm{hist}}\).
    \STATE \quad 3. Define \(\texttt{cdf\_values} = \text{linspace}(1/N_{\mathrm{hist}},1,N_{\mathrm{hist}})\) and build an empirical CDF interpolation \(\texttt{h\_interp}\).
    \STATE \quad 4. Let \(\texttt{deg} = 30\). Fit a polynomial \(\texttt{h\_poly}\) such that
    \STATE \quad \qquad \(\texttt{coeffs} \leftarrow \text{polyfit}(\texttt{u\_train},\, \texttt{h\_train},\, \texttt{deg})\).
    \STATE \quad 5. Define \(\widetilde{h}(x) := h\bigl((x+1)/2\bigr)\).
    
    \vspace{0.5em}
    \STATE \quad \textbf{Compute and store Legendre inner products:}
    \STATE \quad 6. Obtain Legendre-Gauss nodes \((x_g,w_g) \leftarrow \texttt{leggauss}(200)\).
    \STATE \quad 7. For \(n = 0 \ldots \texttt{deg}\):
    \STATE \quad \qquad (a) Evaluate \(P_n(x_g)\).
    \STATE \quad \qquad (b) Evaluate \(\int \widetilde{h}_1(x_g,K)\,P_n(x_g)\,w_g\) and store it.
    
    \vspace{1em}
    \STATE \textbf{Step 2: 3/2 model pricing function + calibration}
    \STATE \quad \textbf{Function} \(\text{price\_vix\_call\_3\_2}(V,K,T,\alpha,\beta,k,r)\):
    \STATE \quad \qquad 1. If \(T \le 10^{-10}\), return \(\max(V - K,\,0)\).
    \STATE \quad \qquad 2. Compute \(p = 1 - e^{-\alpha T}\) and \(\nu = 1 - 2\beta/k^2\).
    \STATE \quad \qquad 3. Compute factor\_out:
    \STATE \quad \qquad \quad \(\displaystyle \frac{2\alpha e^{-rT}}{k^2 p} \exp\Bigl(-\frac{2\alpha e^{-\alpha T}}{k^2 V p}\Bigr)
                   \times V^{-\beta/k^2 + 1/2}
                   \times \exp\Bigl(\alpha T \bigl(-\tfrac{\beta}{k^2}+\tfrac12\bigr)\Bigr).\)
    \STATE \quad \qquad 4. Numerically integrate from \(0\) to \(1/K\).  
    \STATE \quad \qquad 5. Return \(\max(\text{call\_price},\,0)\).
    
    \vspace{1em}
    \STATE \quad \textbf{Calibrate} \(\alpha, \beta, k\):
    \STATE \quad 1. For each data point in \(\texttt{df\_}\):
    \STATE \quad \qquad (a) Evaluate 3/2 model price \(\text{price\_vix\_call\_3\_2}\).
    \STATE \quad \qquad (b) Accumulate squared error between model and market price.
    \STATE \quad 2. Use \(\texttt{minimize}\) to find \(\hat\alpha, \hat\beta, \hat k\) that minimize SSE.
    
    \vspace{1em}
    \STATE \textbf{Step 3: Read option data and train the Legendre + 3/2 models}
    \STATE \quad 1. Read Excel file and obtain historical VIX call option data $\text{VIX}_{t_i}$ and corresponding factor value $x_i$.
    \STATE \quad 2. Fix strike \(K_{\mathrm{example}}=0.20\).
    \STATE \quad 3. \textbf{Legendre calibration}: 
    \STATE \quad \qquad (a) Compute \(\langle \widetilde{h}_1,\,P_n\rangle\) for all \(n\).
    \STATE \quad \qquad (b) Set up \(\texttt{Pn\_matrix}\).  
    \STATE \quad \qquad (c) Define the objective function w.r.t. \(\kappa\) using exponent \(\exp[-\frac12\,\kappa\,n(n+1)\,(T-t)]\).
    \STATE \quad \qquad (d) Optimize to find \(\hat{\kappa}\).
    \STATE \quad 4. \textbf{3/2 model calibration}:
    \STATE \quad \qquad (a) \(\hat\alpha, \hat\beta, \hat k = \text{calibrate\_3\_2\_model}(\texttt{df},\,K_{\mathrm{example}})\).
    
    \vspace{1em}
    \STATE \textbf{Step 4: Compare Legendre vs. 3/2 for VIX data}
    \STATE \quad 1. \textbf{Legendre pricing} for each \(\texttt{df}\) row:
    \STATE \quad \qquad Evaluate \(\sum_{n} \bigl(\tfrac{2n+1}{2}\langle\widetilde{h}_1,P_n\rangle\,P_n(x_i)\bigr) \exp[-\frac12\,\hat{\kappa}\,n(n+1)\,(T-t)]\).
    \STATE \quad \qquad Multiply by \(\exp(-rT)\).
    \STATE \quad 2. \textbf{3/2 pricing} for each \(\texttt{df}\) row:
    \STATE \quad \qquad Evaluate \(\text{price\_vix\_call\_3\_2}(V,K,T,\hat\alpha,\hat\beta,\hat k,r)\).
    \STATE \quad 3. Compute absolute/relative errors and optionally plot results.
  \end{algorithmic}
  }
\end{algorithm}

\end{document}